\documentclass[pra,groupedaddress,superscriptaddress,twocolumn,showpacs,amsmath,amssymb,a4paper]{revtex4}
\usepackage{geometry}
\usepackage{graphicx}
\usepackage{mathrsfs}
\usepackage{amssymb}
\usepackage{amsmath}
\usepackage{amsthm}

\newtheorem{proposition}{Proposition}

\newtheorem{corollary}{Corollary}

\theoremstyle{definition}

\newcommand{\bra}[1]{\langle #1|}
\newcommand{\ket}[1]{| #1 \rangle }

\newcommand{\tr}[1]{{\rm tr}[#1]}
\newcommand{\be}{\begin{eqnarray}}
\newcommand{\ee}{\end{eqnarray}}

\newcommand{\stirling}[2]{\genfrac{\lbrace}{\rbrace}{0pt}{}{#1}{#2}}

\newcommand{\cE}{{\cal E}}

\newcommand{\cD}{{\cal D}}

\newcommand{\cM}{{\cal M}}

\newcommand{\cT}{{\cal T}}

\newcommand{\cS}{{\cal S}}
\newcommand{\cH}{{\cal H}}

\newcommand{\cO}{{\cal O}}

\newcommand{\cP}{{\cal P}}

\begin{document}

\title{Dissociation and annihilation of multipartite entanglement structure \\ in dissipative quantum dynamics}

\author{Sergey N. Filippov}

\affiliation{Research Center for Quantum Information, Institute of
Physics, Slovak Academy of Sciences, D\'{u}bravsk\'{a} cesta 9,
Bratislava 84511, Slovakia}

\affiliation{Department of Experimental Physics, Comenius
University, Mlynsk\'{a} dolina, Bratislava 84248, Slovakia}

\affiliation{Department of Theoretical Physics, Moscow Institute
of Physics and Technology (State University), Institutskii
Pereulok 9, Dolgoprudny, Moscow Region 141700, Russia}

\affiliation{Institute of Physics and Technology, Russian Academy
of Sciences, Nakhimovskii Pr. 34, Moscow 117218, Russia}

\author{Alexey A. Melnikov}

\affiliation{Institute of Physics and Technology, Russian Academy
of Sciences, Nakhimovskii Pr. 34, Moscow 117218, Russia}

\affiliation{Institute for Quantum Optics and Quantum Information,
Austrian Academy of Sciences,\\ Technikerstra{\ss}e 21a, Innsbruck
6020, Austria}

\affiliation{Institute for Theoretical Physics, University of
Innsbruck, Technikerstra{\ss}e 25, Innsbruck 6020, Austria}

\author{M\'{a}rio Ziman}

\affiliation{Research Center for Quantum Information, Institute of
Physics, Slovak Academy of Sciences, D\'{u}bravsk\'{a} cesta 9,
Bratislava 84511, Slovakia}

\affiliation{Faculty of Informatics, Masaryk University,
Botanick\'{a} 68a, Brno 60200, Czech Republic}

\begin{abstract}
We study the dynamics of the entanglement structure of a
multipartite system experiencing a dissipative evolution. We
characterize the processes leading to a particular form of output
system entanglement and provide a recipe for their identification
via concatenations of particular linear maps with
entanglement-breaking operations. We illustrate the applicability
of our approach by considering local and global depolarizing
noises acting on general multiqubit states. A difference in the
typical entanglement behavior of systems subjected to these noises
is observed: the originally genuine entanglement dissociates by
splitting off particles one by one in the case of local noise,
whereas intermediate stages of entanglement clustering are present
in the case of global noise. We also analyze the definitive phase
of evolution when the annihilation of the entanglement compound
finally takes place.
\end{abstract}

\pacs{03.67.Bg, 03.65.Ud, 03.65.Yz, 03.67.Mn}

\maketitle

\section{Introduction}

The physical phenomenon of entanglement naturally appears in
composite quantum systems via interactions among constituents.
Simple collision models already teach us that different
interaction types lead to various types of multipartite
entanglement \cite{ziman-buzek-2011}. Systems with local
Hamiltonians exhibit correlations between the degree of
entanglement and eigenenergies \cite{li-haldane-2008,pizorn-2013},
phase transitions \cite{osborne-2002,osterloh-2002,hofmann-2013},
and the number of interacting bodies \cite{levi-mintert-2013}.
Multipartite entanglement finds uses in quantum networking
applications such as secret sharing \cite{hillery-1999}, secret
voting \cite{hillery-2006}, open-destination teleportation
\cite{karlsson-1998}, etc. For the latter purposes, entanglement
can be created within the system not only by interaction among
constituent bodies but also by a properly engineered interaction
with the environment \cite{palma-1989,marr-2003,lucas-2013}.

Suppose the prepared multipartite entangled state is intended for
use in an entanglement-enabled quantum protocol involving remote
clients. While transferring the quantum information to recipients,
the state will be modified by inevitable noise processes. It can
happen that the type of multipartite entanglement received by the
clients differs significantly from the original one, and the
realization of the desired protocol becomes impossible. Similarly,
uncontrollable noise processes in quantum memory devices can
result in destroying particular correlations within the stored
multipartite system and make the released state ineffective
\cite{pastawski-2009,simon-2010}. Degradation of entanglement also
imposes limitations on the benefit of advanced quantum metrology
relying on genuinely multipartite entangled states
\cite{giovannetti-2011}. These examples demonstrate the necessity
of tracking the multipartite entanglement dynamics and finding
noise levels corresponding to the change of entanglement type.

Previous efforts in this direction relied on specific entanglement
measures. Negativity \cite{vidal-2002} --- a measure detecting
negativity of the density matrix under partial transpose (NPT)
\cite{peres-1996} --- was originally used by Simon \textit{et al}.
\cite{simon-2002} and D\"{u}r \textit{et al}. \cite{dur-2004} to
analyze GHZ, W, and cluster states under local depolarizing noise.
Then Bandyopadhyay \textit{et al}. \cite{bandyopadhyay-2005} and
Hein \textit{et al}. \cite{hein-2005} utilized it to study the
behavior of GHZ states and graph states, respectively, under
general local homogeneous noise. Generalized GHZ-type states under
a local amplitude-damping channel were considered with the help of
negativity by Man \textit{et al}. \cite{man-2008}. Aolita
\textit{et al}. exploited negativity to study effects of local
depolarizing, dephasing, and generalized amplitude damping
channels on GHZ states \cite{aolita-2008} and graph states
\cite{aolita-2010}. Those results were obtained for an arbitrary
number of qubits (except for some graph states \cite{aolita-2010}
and randomly sampled states \cite{aolita-2009}) due to the
ultimate simplicity of negativity computation. Depolarization and
dephasing of qudit GHZ states were considered via negativity in
\cite{liu-2009}. Similarly, concatenated GHZ states (where blocks
of a small number of qubits are GHZ states themselves) were
considered in \cite{frowis-2011}. However, the negativity does not
provide comprehensive information about the entanglement structure
because it can be sensitive to the entanglement with respect to a
particular bipartition only (remember, e.g., bound-entangled PPT
states \cite{horodecki-1998} and biseparable but non-triseparable
states \cite{bennet-1999}).

The absence of full separability can also be detected by some
other measures. For instance, Carvalho \textit{et al}. used the
lower bound for a specific generalization of the concurrence and
applied it to the dynamics of several-qubit GHZ and W states under
amplitude-damping and dephasing local channels
\cite{carvalho-2004}. G\"{u}hne \textit{et al}. used the geometric
measure of entanglement \cite{barnum-2001,wei-2003} to study
global dephasing process of four-qubit GHZ, cluster, W, and Dicke
states \cite{guhne-2008}. Grimsmo \textit{et al}. used the
entropic measure for average $n$-partite entanglement over quantum
trajectories~\cite{grimsmo-2012}. Gheorghiu \textit{et al}.
developed the evolution of an averaged SL-invariant entanglement
measure for local decoherence \cite{gheorghiu-2012}. A similar
approach with a lower bound of the concurrence was exploited in
\cite{man-2012}. A non-zero value of these quantities indicates
the presence of some entanglement within the quantum system, but
gives little information about its particular form and, therefore,
the benefit of this entanglement for some applications remains
questionable. Moreover, vanishing values of the above measures
cannot guarantee the full separability of the state, and thus, the
problem of fundamental noise limits eliminating any form of
entanglement (resulting in fully separable states) is still open.

Genuine multipartite entanglement is the exact opposite of full
separability: this form of entanglement is intrinsically
multiparticle and cannot be attributed to the entanglement
distributed among smaller subsystems. The detection of genuine
entanglement for specific quantum states has been a subject of
intensive recent research (see, e.g.,
\cite{guhne-2010,huber-2010,devicente-2011,jungnitsch-2011,novo-2013,huber-2013}
and references therein). Dissipative evolution of genuine
multipartite entanglement has been analyzed with the help of some
measures. The mean value of a projector-like witness
\cite{bourennane-2004} was used by Bodoky \textit{et al}. to study
several qubits within a heuristic model of decoherence based on
local relaxation and dephasing times \cite{bodoky-2009}. Campbell
\textit{et al}. used fidelity- and collective spin-based
entanglement witnesses to analyze the dynamics of genuine
multipartite entanglement of Dicke states under local
amplitude-damping, phase-damping, and depolarizing channels
\cite{campbell-2009}. Tripartite negativity and generalized
concurrence were also applied to the dissipative dynamics of GHZ
and W three-qubit entangled states
\cite{weinstein-2009,altintas-2010,an-2011,siomau-2012,ryu-2012,buscemi-2013}.
Let us recall that the above measures are not precise, i.e., their
zero values do not imply in general that the genuine entanglement
is lost. On the other hand, precise measures (based on the convex
roof definitions) are quite hard to compute. This is the main
reason why the research in entanglement dynamics is usually
restricted to particular initial states (GHZ, W, X, Dicke, etc.)
and the use of relatively simple measures.

Despite existing results for noises preserving genuine
entanglement and entanglement on the whole (absence of full
separability), the evolution of entanglement structure still
remains unexplored.  The aim of this paper is to track the
transformations of entanglement structure during dissipative
processes. By ``structure'' we understand the number of separate
components and the number of particles within each of them (with
allowance for convex mixtures)
\cite{levi-mintert-2013,dur-1999,seevinck-2008}. This structure
resembles a Russian nested doll, and dissipative evolution maps
states from the outer to the inner dolls. The evolution of
entanglement structure can be seen as a dissociation of the
entanglement compound due to interaction with the ``solvent''
(particles of the environment). Note that the ``entanglement
compound'' refers to a genuinely entangled multipartite component
and differs from the concept of an ``entanglement molecule'' whose
bonds depict entanglement of reduced two-particle states
\cite{dur-2001}. The idea of tracking the entanglement structure
was realized for three-qubit GHZ states under global
depolarization in \cite{eltschka-2012} and for the restricted
Hilbert space of single-excitation states in
\cite{lougovski-2009}. We do not restrict ourselves to particular
input states and develop a theory of transformations that map
\textit{any} initial state into a chosen doll. Note that
mainstream research is focused on showing that a particular state
is outside a given doll (mostly that of biseparable states)
\cite{levi-mintert-2013,guhne-2010,huber-2010,devicente-2011,jungnitsch-2011,novo-2013,huber-2013,seevinck-2008},
whereas ours ensures the opposite and matches the recent approach
of Ref. \cite{kampermann-2012}. Our methodology relies on a neat
decomposition of the physical map into simpler (but not necessary
physical) processes involving entanglement-breaking operations
\cite{horodecki-2003}. The criteria obtained are formulated for
general quantum channels.

To illustrate our approach, we discuss examples of local and
global depolarizing noises modelling individual and common baths,
respectively. Local depolarizing noises are relevant in quantum
communication tasks (exploiting, e.g., optical fibers) as well as
in purely physical systems such as nuclear spins in molecules
\cite{emerson-2007}. Global depolarizing noise is an appropriate
model in experiments where full-rank quantum states are detected
\cite{barreiro-2010,lavoie-2010} and is argued to be the
worst-case scenario of system-environment interactions
\cite{dur-2005}. We find the noise levels of corresponding
entanglement structure dissociations and reveal differences in the
typical dissociation behavior between local and global noises.

The paper is organized as follows.

In Sec. \ref{section-formalism}, we precisely describe the
multipartite entanglement formalism used, with attention being
paid to higher order partitions (tripartitions, tetrapartitions,
etc.) which are often omitted from consideration. In Sec.
\ref{section-dynamics}, we recall the necessary information about
general and local quantum channels. In Sec. \ref{section-ea}, the
problem under investigation (dynamics of entanglement structure)
is precisely formulated. In Sec. \ref{section-methodology}, we
accomplish the development of methodology and derive the criteria
of entanglement dissociation and annihilation. In Sec.
\ref{section-application}, we provide a recipe for applying the
obtained criteria to the above-mentioned noises. In Sec.
\ref{section-discussion}, the physical meaning of the results is
discussed. In Sec. \ref{section-summary}, we concisely summarize
the ideas, methods, and achieved results.

\section{\label{section-formalism} Multipartite entanglement formalism}

To express the idea of entanglement structure quantitatively, one
can make use of the following formalism.

Whenever we speak about entanglement, we imply a particular
partition of the composite system. In general, an $N$-body system
$ABC\ldots$ can be partitioned into $k$ subsystems, where $k$
ranges from $2$ to $N$. If the system is not partitioned at all,
we will reckon $k=1$. One can divide the $N$-body system
$ABC\ldots$ into $k$ subsystems (also referred to as parties) in
$\stirling{N}{k}$ different ways, where $\stirling{N}{k} =
\frac{1}{k!} \sum_{m=0}^k (-1)^m \binom{k}{m} (k-m)^N$ is the
Stirling number of the second kind. Denote by $\cP^k$ a set of
possible partitions into $k$ parties. Partitions are ordered in
such a way that the parts with fewer bodies go first. Then, for a
three-body system $ABC$ we have $\cP^1(ABC) = \{ABC\}$,
$\cP^2(ABC) = \{A|BC,B|AC,C|AB\}$, and $\cP^3(ABC) = \{A|B|C\}$.
In the case of a four-body system $ABCD$, the sets of possible
partitions are $\cP^1(ABCD) = \{ABCD\}$, $\cP^2(ABCD) = \{A|BCD$,
$B|ACD$, $C|ABD$, $D|ABC$, $AB|CD$, $AC|BD$, $AD|BC\}$,
$\cP^3(ABCD) = \{A|B|CD$, $A|C|BD$, $A|D|BC$, $B|C|AD$, $B|D|AC$,
$C|D|AB\}$, and $\cP^4(ABCD) = \{A|B|C|D\}$. Denote by $\cP_j^k$
the $j$-th partition of the set $\cP^k$, e.g., $\cP_5^3(ABCD) =
B|D|AC$. In order to address the $m$th subsystem of the partition
$\cP_j^k$, we will use the notation $[\cP_j^k]_m$, e.g.,
$[\cP_5^3(ABCD)]_2 = D$.

Quantum states of the system $ABC\ldots$ are described by density
operators $\varrho^{ABC\ldots}$ (positive and with unit trace)
acting on the Hilbert space $\cH^{ABC\ldots} \equiv
\cH^{A}\otimes\cH^{B}\otimes\cH^{C}\otimes\ldots$ and altogether
forming the convex set $\cS(\cH^{ABC\ldots})$. A state $\varrho$
is called separable with respect to a particular partition
$\cP_j^k$ if the resolution $\varrho = \sum_i \mu_i
\varrho_i^{[\cP_j^k]_1} \otimes \cdots \otimes
\varrho_i^{[\cP_j^k]_k}$ holds true for some probability
distribution $\{\mu_i\}$ and density operators
$\varrho_i^{[\cP_j^k]_m}$, $m=1,\ldots,k$. We will denote such a
separable state as $\sigma_j^k$ for brevity. If $\varrho \ne
\sigma_j^k$ for any $\sigma_j^k$, then $\varrho$ is said to be
entangled with respect to the partition $\cP_j^k$.

The above consideration of partitions is important because the
physics of multipartite entanglement can be quite
counterintuitive. For instance, the three-qubit state of Refs.
\cite{bennet-1999,divincenzo-2003} is separable with respect to
any bipartition $\cP_j^2$ but is entangled with respect to
tripartition $\cP^3$. Another example is a four-qubit Smolin state
\cite{smolin-2001} which is separable with respect to bipartitions
$\cP_5^2$, $\cP_6^2$, $\cP_7^2$ and is entangled with respect to
bipartitions $\cP_1^2$, $\cP_2^2$, $\cP_3^2$, $\cP_4^2$, any
tripartition $\cP_j^3$, and quartering $\cP^4$.

Now we can define the concept of $k$-separability of a quantum
state, which indicates that the state can accommodate components
each of which has $k$ separate parties. Namely, the state
$\varrho$ is called $k$-separable and denoted
$\varrho_{k\text{-sep}}$ if it adopts the resolution
$\sum_{j=1}^{\stirling{N}{k}} p_j^k \sigma_j^k$ for some
probability distribution $\{p_j^k\}_j$ and separable density
operators $\sigma_j^k$. Note that $\varrho_{k\text{-sep}}$ can
still be entangled with respect to partitions $\cP_j^k$ if
$\stirling{N}{k}
> 1$. Clearly, if the state is $k$-separable, then it is also
$(k-1)$-separable, which implies the inclusion relation
$\cS_{N\text{-sep}} \subset \ldots \subset \cS_{2\text{-sep}}
\subset \cS_{1\text{-sep}}$ for convex sets of $k$-separable
states. A natural measure of separability appears:
\begin{equation}
\label{K-separability} K_{\rm sep} [\varrho] :=
\max\limits_{\varrho = \varrho_{k\text{-sep}}} k.
\end{equation}

\noindent If $K_{\rm sep} [\varrho] = 1$, then the state $\varrho$
is called genuinely entangled (GE). If $K_{\rm sep} [\varrho] =
N$, then the state $\varrho$ is fully separable (FS).

One can quantify multipartite entanglement in an alternative way
by counting the number of bodies that are actually
entangled~\cite{levi-mintert-2013,seevinck-2008,vedral-1997,gisin-1998}.
This number would indicate the resources needed to create the
state. For instance, the state $\varrho^{AB} \otimes
\varrho^{CDE}$ of a 5-body system $ABCDE$ is 2-separable but
comprises a party $CDE$ which can be genuinely entangled ($K_{\rm
sep}[\varrho^{CDE}]=1$), i.e., requires 3 bodies to be entangled.
To embody this idea in a precise manner, we introduce the
following definition of resource-intensiveness (compatible with
the concepts of entanglement depth~\cite{sorensen-2001} and
producibility~\cite{guhne-toth-briegel-2005}):
\begin{equation}
\label{R-entanglement}  R_{\text{ent}} [\varrho] := \!\!\!
\min\limits_{\varrho = \sum\limits_{k=1}^N
\sum\limits_{j=1}^{\stirling{N}{k}} p_j^k \sigma_j^k}
\max\limits_{m=1,\ldots,k} \! \left\{ \#~\text{bodies within}~
[\cP_j^k]_m \right\} \! .
\end{equation}

Denote by $\cS_{r\text{-ent}} = \{ \varrho : R_{\rm ent} [\varrho]
\le r\}$ the convex set of $r$-entangled states. Obviously,
$\cS_{1\text{-ent}} \subset \cS_{2\text{-ent}} \subset \ldots
\subset \cS_{N\text{-ent}}$. Importantly, $\cS_{1\text{-ent}} =
\cS_{N\text{-sep}}$, $\cS_{(N-1)\text{-ent}} =
\cS_{2\text{-sep}}$, and $\cS_{N\text{-ent}} = \cS_{1\text{-sep}}
= \cS(\cH^{ABC\ldots})$. Depending on the quantum state, the range
of $R_{\text{ent}}$ can be $\lceil \frac{N}{K_{\text{sep}}} ,
N-K_{\text{sep}}+1]$ for a fixed $K_{\text{sep}}$, and the range
of $K_{\text{sep}}$ can be $\lceil \frac{N}{R_{\text{ent}}} ,
N-R_{\text{ent}}+1]$ for a fixed $R_{\text{ent}}$
\footnote{Hereafter, $\lceil x \rceil$ denotes the smallest
integer greater than or equal to $x$, and $\lfloor x \rfloor$
denotes the greatest integer less than or equal to $x$.}. The
relations between two families of sets $\{\cS_{k\text{-sep}}\}$
and $\{\cS_{r\text{-ent}}\}$ for a four-body system are shown in
Fig.~\ref{figure1}.

\begin{figure}
\includegraphics[width=8.5cm]{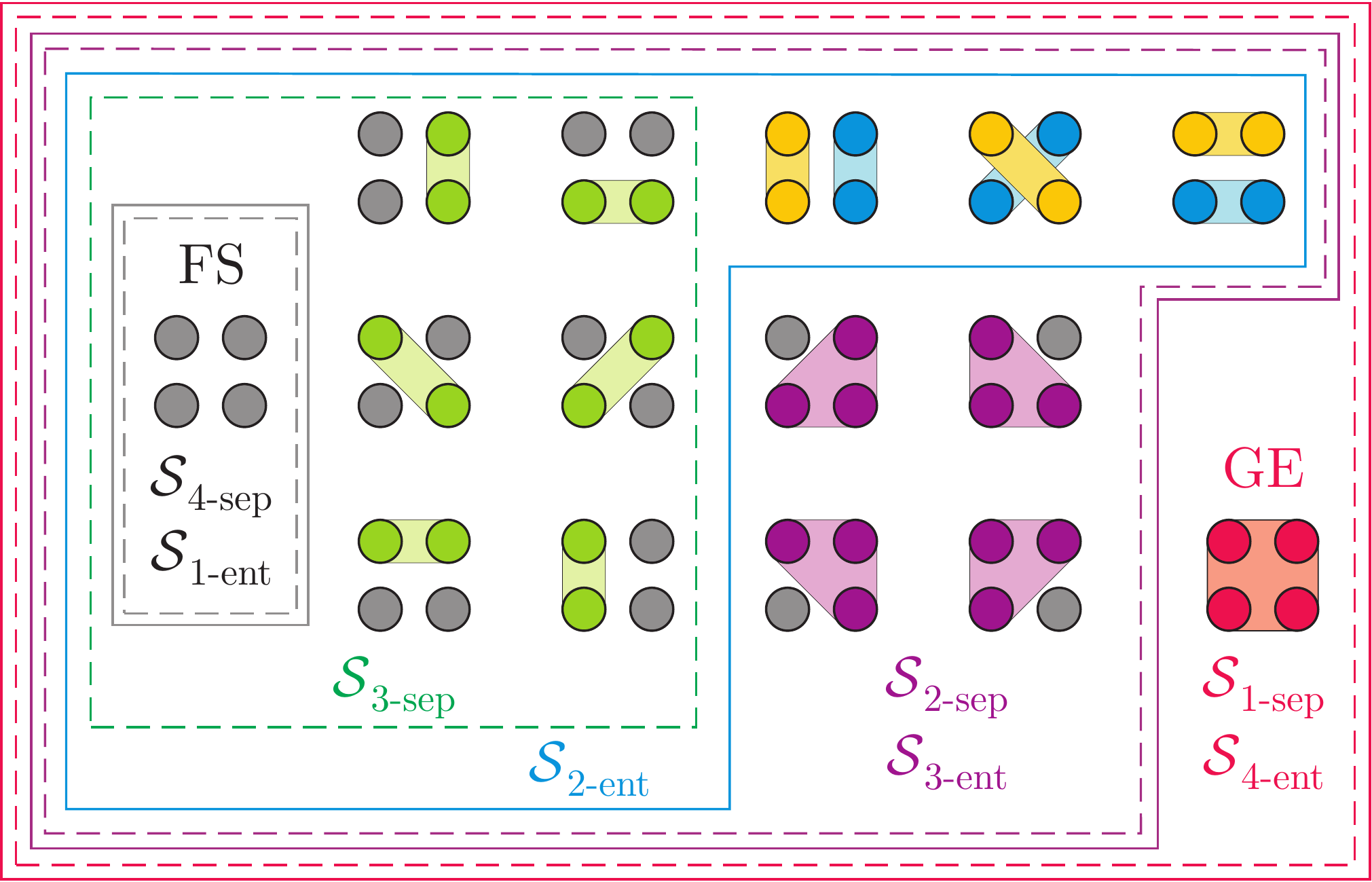}
\caption{\label{figure1} Schematic of sets $\cS_{k\text{-sep}}$
(dashed) and $\cS_{r\text{-ent}}$ (solid) for a four-body system.}
\end{figure}

\section{\label{section-dynamics} Quantum dynamics}
We describe the physical evolution of open quantum systems by the
input--output formalism of quantum channels: $\varrho_{\text{out}}
= \Phi[\varrho_{\text{in}}]$, where $\Phi:\cT(\cH_{\rm
in})\to\cT(\cH_{\rm out})$ is a completely positive
trace-preserving (CPT) linear map on trace-class operators
$\cT(\cH_{\rm in})$. The physical meaning of the evolution via a
CPT map $\Phi$ can be readily seen from the Stinespring
dilation~\cite{stinespring}: $\Phi[\varrho_{\text{in}}] \equiv
{\rm tr}_{\rm env} [U (\varrho_{\rm in} \otimes \xi_{\rm env})
U^{\dag}]$ for some state of the environment $\xi_{\rm env}$ and
some unitary operator $U\in\cT(\cH_{\rm in}\otimes\cH_{\rm env})$.
Complete positivity (CP) of the map $\Phi$ acting on a system $S$
guarantees that $(\Phi^{S}\otimes{\rm Id}^{\rm
anc})[\varrho^{S+\text{anc}}] \ge 0$ for all composite states
$\varrho^{S+\text{anc}}\in\cS(\cH^{S+\text{anc}})$ of the system
$S$ and an arbitrary ancilla, with ${\rm Id}$ being the identity
transformation. Equivalently, the map $\Phi$ is CP if it adopts
the diagonal sum representation $\Phi[X] = \sum_k A_k X
A_k^{\dag}$. If the Kraus operators $A_k:\cH_{\rm in} \mapsto
\cH_{\rm out}$ satisfy $\sum_k A_k^{\dag} A_k = I_{\rm in}$
(identity operator), then $\Phi$ is CPT.

In order to define a linear map $\Phi$ acting on a system $S$, we
will use the Choi--Jamio{\l}kowski
isomorphism~\cite{choi-1975,jamiolkowski-1972}:
\begin{eqnarray}
&& \label{choi-matrix} \Omega_{\Phi}^{SS'} := (\Phi^S \otimes {\rm
Id}^{S'})[\ket{\Psi_+^{SS'}}\bra{\Psi_+^{SS'}}], \\
&& \label{map-through-choi} \Phi [X] = d^S \, {\rm tr}_{S'} \, [
\,\Omega_{\Phi}^{SS'} (I_{\rm out}^S \otimes X^{\rm T}) \, ],
\end{eqnarray}

\noindent where $d={\rm dim}\cH$, $\ket{\Psi_+^{SS'}} =
(d^S)^{-1/2} \sum_{i=1}^{d^S} \ket{i\otimes i'}$ is a maximally
entangled state shared by system $S$ and its clone $S'$, $X^{\rm
T}=\sum_{i,j} \bra{j} X \ket{i} \ket{i'}\bra{j'} \in \cT(\cH_{\rm
in}^{S'})$ is the transposition in some orthonormal basis, and
${\rm tr}_{S'}$ denotes the partial trace over $S'$. The linear
map $\Phi^S$ is CP if and only if $\Omega_{\Phi}^{SS'} \ge 0$.

Since our main interest is focused on many-body systems, let us
consider a composite system $S=ABC\ldots$ acted upon by some
channel $\Phi^{S}$. To begin with, $\ket{\Psi_+^{SS'}} = (d^A d^B
d^C \cdots )^{-1/2} \sum\limits_{i=1}^{d^A}
\sum\limits_{j=1}^{d^B} \sum\limits_{k=1}^{d^C}
\sum\limits_{\cdots} \ket{ijk\cdots}\otimes \ket{i'j'k'\cdots} =
\ket{\Psi_+^{AA'}} \otimes \ket{\Psi_+^{BB'}} \otimes
\ket{\Psi_+^{CC'}} \otimes \cdots$, which explicitly shows the
separability of the maximally entangled state with respect to the
partition $AA'|BB'|CC'|\ldots$. While constructing the Choi
operator (\ref{choi-matrix}), the map $\Phi^{ABC\ldots}$ can in
general entangle these subsystems.

Suppose a local channel $\Phi_1^A \otimes \Phi_2^B \otimes
\Phi_3^C \otimes \cdots$ which serves as an adequate model in
situations when each particle is sent to a corresponding receiver
through an individual quantum cable (Fig.~\ref{figure2}a). In this
case, $\Omega_{\Phi_1 \otimes \Phi_2 \otimes \Phi_3 \otimes
\ldots}^{ABC \ldots A'B'C' \ldots} = \Omega_{\Phi_1}^{AA'} \otimes
\Omega_{\Phi_2}^{BB'} \otimes \Omega_{\Phi_3}^{CC'} \otimes
\cdots$. Clearly, $\Phi_1^A \otimes \Phi_2^B \otimes \Phi_3^C
\otimes \cdots$ is CP if and only if each of the maps $\Phi_1^A$,
$\Phi_2^B$, $\Phi_3^C$, $\ldots$ is CP.

\begin{figure}
\includegraphics[width=8.5cm]{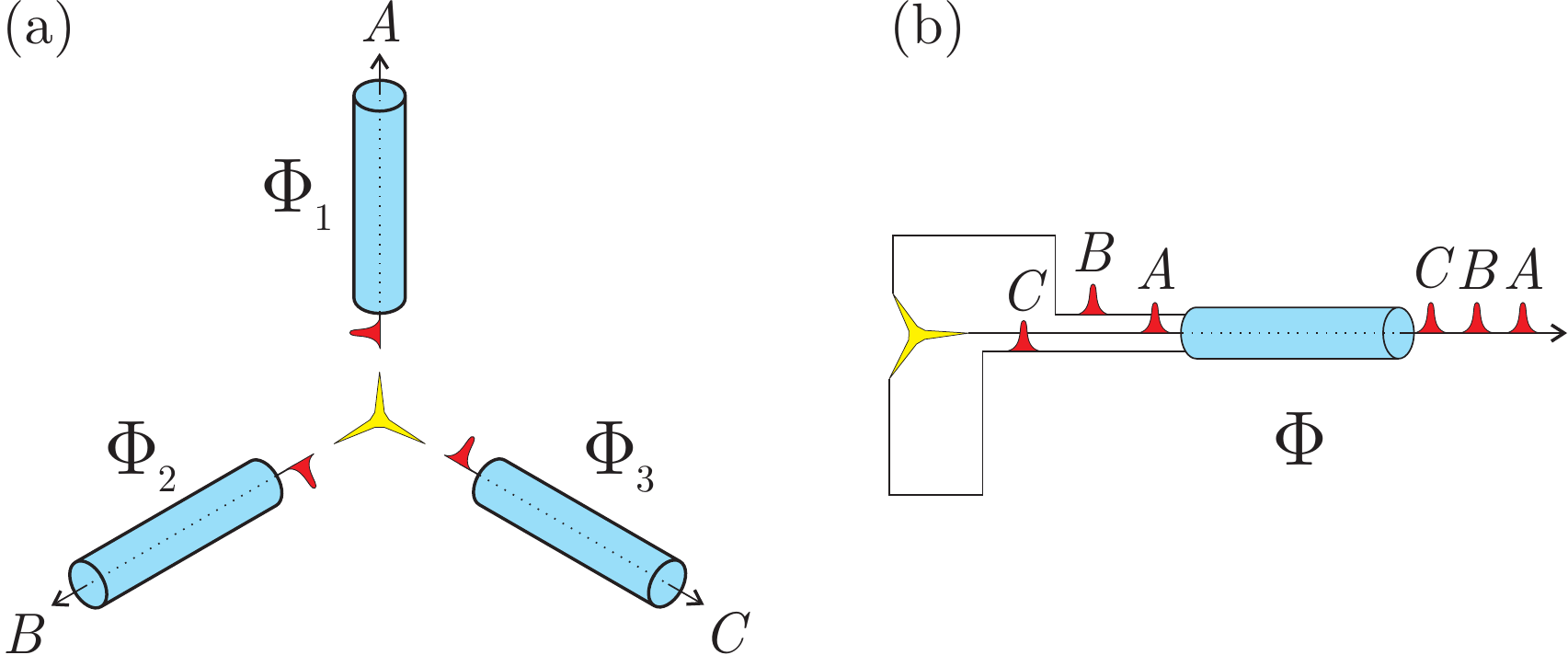}
\caption{\label{figure2} Local channels: (a) general; (b)
homogeneous.}
\end{figure}
In quantum communication, the typical scenario is to use a single
quantum cable to transmit time-separated parties of a multipartite
state from the encoder to the decoder (Fig.~\ref{figure2}b).
Neglecting the memory effects, the evolution of a multipartite
system is governed by the homogeneous local channel $\Phi^{\otimes
N}$, which also appears in the definition of channel capacities
(see, e.g., the review~\cite{holevo-giovannetti}).

\section{\label{section-ea} Problem formulation}

Consider a composite $N$-body system $S=ABC\ldots$ that undergoes
the physical evolution $\varrho_{\rm out} = \Phi [\varrho_{\rm
in}]$ determined by some CPT map $\Phi$ (we also assume $\cH_{\rm
in}^S = \cH_{\rm out}^S$). If $\varrho_{\rm out}$ is separable
with respect to the partition $\cP_j^k$ (i.e. $\varrho_{\rm
out}=\sigma_j^k$), then we say that the channel $\Phi$ dissociates
the entanglement compound of a given $\varrho_{\rm in}$ into
smaller compounds of $[\cP_j^k]_1,\ldots,[\cP_j^k]_k$ and denote
by $\cD_j^k (\varrho_{\rm in})$ the set of such channels. If the
channel $\Phi$ dissociates the entanglement of all input states
$\varrho_{\rm in} \in \cS(\cH^S)$ in this way, then we will refer
to $\Phi$ as dissociating entanglement with respect to the
partition $\cP_j^k$ and denote $\Phi \in \cD_j^k \equiv
\bigcap_{\varrho_{\rm in} \in \cS(\cH^S)} \cD_j^k (\varrho_{\rm
in})$.

Using entanglement measures (\ref{K-separability}) and
(\ref{R-entanglement}), we can quantitatively describe the
processes of entanglement structure dynamics. Namely, denote by
$k\text{Sep}(\varrho_{\rm in})$ a set of channels $\Phi$ such that
$K_{\rm sep} \left[ \Phi[\varrho_{\rm in}] \right] \ge k$. By
construction, $k\text{Sep} (\varrho_{\rm in})$ is a convex hull of
the sets $\cD_j^k (\varrho_{\rm in})$. Similarly, $r\text{Ent}
(\varrho_{\rm in})$ is a set of channels $\Phi$ such that $R_{\rm
ent} \left[ \Phi[\varrho_{\rm in}] \right] \le r$. Regarding
state-independent properties, we straightforwardly introduce the
sets of channels $k\text{Sep} := \bigcap_{\varrho_{\rm in} \in
\cS(\cH^S)} k\text{Sep} (\varrho_{\rm in})$ and $r\text{Ent} :=
\bigcap_{\varrho_{\rm in} \in \cS(\cH^S)} r\text{Ent}
(\varrho_{\rm in})$. The developed formalism of
Sec.~\ref{section-formalism} immediately results in the following
inclusion diagram for the above sets:
\begin{equation}
\begin{array}{ccccccccc}
  N\text{Sep} & \subset & (N-1)\text{Sep} &\subset & \cdots & \subset & 2\text{Sep} & \subset & 1\text{Sep} \\
  \parallel &   &   &   &   &   & \parallel & & \parallel \\
  1\text{Ent} & \subset & 2\text{Ent} & \subset & \cdots & \subset & (N-1)\text{Ent} & \subset & N\text{Ent} \\
  \parallel &   &   &   &   &   & \parallel & & \parallel \\
  \text{EA} &   &   &   &   &   & \text{DGE} & & \text{CPT} \\
\end{array} \nonumber
\end{equation}

\noindent We have used a special notation for two distinctive
classes of channels:
\begin{itemize}
\item entanglement annihilating channels (EA) transforming any
input state into a fully separable one
\cite{moravcikova-ziman-2010};

\item channels that dissociate genuine entanglement (DGE), thus,
transforming genuinely entangled states into non-genuinely
entangled ones.
\end{itemize}

The problem under investigation is twofold: (i) to characterize
the sets of channels $k\text{Sep}(\varrho_{\rm in})$ and
$r\text{Ent}(\varrho_{\rm in})$ as well as state-independent sets
from the above diagram, (ii) to track how exactly the
multiparticle entanglement structure dissociates under particular
noises. Our special attention is paid to EA and DGE channels.

Before proceeding to the derivation of criteria, we need to
clarify the relation between the problem involved and the well
known approaches developed so far.

Consider a (not necessarily composite) system $S$ acted upon by a
channel $\Phi:\cT(\cH_{\rm in}^S) \mapsto \cT(\cH_{\rm out}^S)$.
If the Choi operator $\Omega_{\Phi}^{SS'}$ is separable with
respect to the partition $S|S'$, then $\Phi$ is a so-called
entanglement-breaking (EB) map \cite{horodecki-2003,holevo-2008},
whose peculiarity is that $(\Phi^S\otimes{\rm Id}^{\rm
anc})[\varrho^{S+{\rm anc}}]$ is separable with respect to the
partition $S|{\rm anc}$ for all density operators $\varrho^{S+{\rm
anc}} \in \cS(\cH^{S+\text{anc}})$. In fact, separability of
$\Omega_{\Phi}^{SS'}$ implies that $\Phi$ has the Holevo form
$\Phi[X] = \sum_k \tr{F_k X} \omega_k$, where $\{F_k\}$ is a
positive operator-valued measure and $\omega_k \in \cS(\cH_{\rm
out})$, i.e. $\Phi$ is a measure-and-prepare procedure. The latter
representation, in its turn, implies \cite{horodecki-2003} that
there exists a diagonal sum representation with rank-1 Kraus
operators $A_k \propto \ket{\varphi_k} \bra{\psi_k}$ with
$\ket{\psi_k}\in\cH_{\rm in}^S$ and $\ket{\varphi_k}\in\cH_{\rm
out}^S$.

As concerns a composite system $S=ABC\ldots$, the EB channel
$\Phi^{S}$ disentangles $S$ from any other system but can in
principle result in any entanglement dynamics within $S$ (among
$A$, $B$, $C$, $\ldots$). For instance, the output state can be
genuinely entangled or fully separable depending on the
entanglement of vectors $\ket{\varphi_k}$ constituting Kraus
operators. However, the local channel $\Phi^S = \Phi_1^A \otimes
\Phi_2^B \otimes \Phi_3^C \otimes \cdots$ is entanglement breaking
if and only if each of the channels $\Phi_1^A$, $\Phi_2^B$,
$\Phi_3^C$, $\ldots$ is entanglement breaking. This can be readily
seen from the requirement of separability of the Choi operator
$\Omega_{\Phi_1 \otimes \Phi_2 \otimes \Phi_3 \otimes \ldots}^{ABC
\ldots A'B'C' \ldots} = \Omega_{\Phi_1}^{AA'} \otimes
\Omega_{\Phi_2}^{BB'} \otimes \Omega_{\Phi_3}^{CC'} \otimes
\cdots$ with respect to the partition $ABC\ldots|A'B'C'\ldots$.
Thus, the local entanglement breaking channel is automatically
entanglement annihilating but the converse is not true. These and
other differences between entanglement breaking and entanglement
annihilating channels are discussed in
\cite{moravcikova-ziman-2010,filippov-rybar-ziman-2012,filippov-ziman-2013}.

\begin{figure*}
\includegraphics[width=18cm]{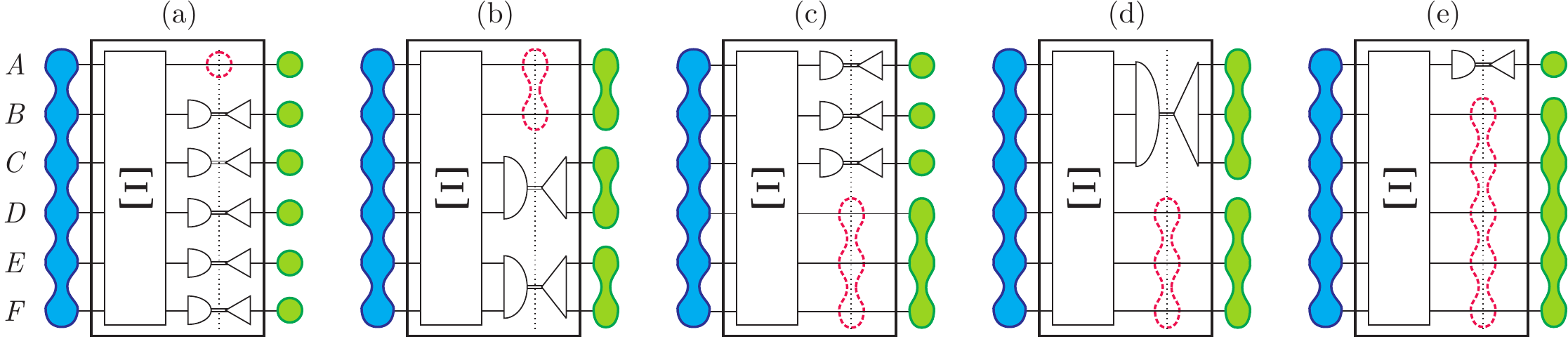}
\caption{\label{figure3} Elementary blocks of entanglement
dissociation for the 6-body system $ABCDEF$ constructed via
concatenation of a linear Hermitian map $\Xi$ and
measure-and-prepare (EB) operations. Semicircles and triangles
depict projections onto $\ket{\psi_n}$ and preparations of
$\ket{\varphi_n}$ of Kraus operators
$A_n\propto\ket{\varphi_n}\bra{\psi_n}$, respectively, and double
lines depict the classical information transfer. Only one
restriction is imposed: $\Xi[\varrho_{\rm in}]$ becomes
positive-semidefinite after performing the ``measure''-part of EB
operations (red dotted compound). Partitions: (a) $A|B|C|D|E|F$,
(b) $AB|CD|EF$, (c) $A|B|C|DEF$, (d) $ABC|DEF$, (e) $A|BCDEF$.}
\end{figure*}

\section{\label{section-methodology} Methodology and criteria}

In this section, we provide criteria to detect the different kinds
of entanglement dissociation discussed above. We start with a
description of our methodology which is based on an extensive use
of various convex sets of operators and maps.

In addition to quantum states described by positive semidefinite
unit trace operators $\varrho\in\cS(\cH^{ABC\ldots})$, an
important role will be played by block-positive operators
\cite{jamiolkowski-1972}. The operator $\xi_j^k$ is called
block-positive with respect to the partition $\cP_j^k$ if it
satisfies
$$\bra{\psi_1^{[\cP_j^k]_1} \otimes \cdots \otimes
\psi_k^{[\cP_j^k]_k}} \xi_j^k \ket{\psi_1^{[\cP_j^k]_1} \otimes
\cdots \otimes \psi_k^{[\cP_j^k]_k}} \ge 0$$ for all vectors
$\psi_1,\ldots,\psi_k$. Block-positive operators are closely
related to entanglement witnesses
\cite{horodecki-1996,horodecki-2001} and can be used to determine
separability: a state $\varrho\in\cS(\cH^{ABC\ldots})$ is
separable with respect to the partition $\cP_j^k$ if and only if
$\tr{ \varrho \xi_j^k } \ge 0$ for all block-positive operators
$\xi_j^k$.

We must emphasize that the concepts of entanglement dissociation
and annihilation from Sec. \ref{section-ea} do not imply any
ancillary system besides the multipartite system $S=ABC\ldots$
itself. This allows to relax CPT condition of the physical
transformation $\Phi$ and construct an extended set $\cE[\Phi]$ of
(mathematical) linear maps $\Upsilon$ having the same entanglement
behavior as $\Phi$ on the corresponding domain of input states.
For example, the extended set $\cE[\cD_j^k(\varrho_{\rm in})]$
consists of linear maps $\Upsilon$ satisfying the only restriction
that $\Upsilon[\varrho_{\rm in}]$ is equal to some $\sigma_j^k$.
Similarly, $\cE [k\text{Sep}(\varrho_{\rm in})]$ and $\cE
[r\text{Ent}(\varrho_{\rm in})]$ denote the extensions of sets
$k\text{Sep}(\varrho_{\rm in})$ and $r\text{Ent}(\varrho_{\rm
in})$, respectively. As we show later, the extensions turn out to
be useful because they adopt a good characterization. The original
set of maps can be found by intersecting with CPT maps, e.g.,
$\cD_j^k(\varrho_{\rm in}) = {\rm CPT} \cap
\cE[\cD_j^k(\varrho_{\rm in})]$.

\begin{proposition} \label{proposition-1}
Suppose a linear map $\Upsilon$ acting on a system $ABC\ldots$.
Then $\Upsilon \in \cE[\cD_j^k(\varrho_{\rm in})]$ if and only if
${\rm tr} \left[ \Omega_{\Upsilon}^{ABC\ldots A'B'C'\ldots} \left(
(\xi_j^k)^{ABC\ldots} \otimes (\varrho_{\rm in}^{\rm
T})^{A'B'C'\ldots} \right) \right] \ge 0$ for all $\xi_j^k$.
\end{proposition}
\begin{proof}
Separability of $\Upsilon[\varrho_{\rm in}]$ with respect to the
partition $\cP_j^k$ is equivalent to the inequality ${\rm tr}
\left[ \Upsilon[\varrho_{\rm in}] \xi_j^k \right] \ge 0$ for all
$\xi_j^k$. Substituting (\ref{map-through-choi}) for
$\Upsilon[\varrho_{\rm in}]$ concludes the proof.
\end{proof}

As a result, the cone $\cE[\cD_j^k (\varrho_{\rm in})]$ is dual to
the cone of maps $\Upsilon^\circ [X] = \xi_j^k \tr{\varrho_{\rm
in} X}$. As concerns the state-independent property $\cD_j^k$, the
map $\Upsilon$ belongs to the set $\cD_j^k$ if its Choi matrix
satisfies ${\rm tr} \left[ \Omega_{\Upsilon}^{ABC\ldots
A'B'C'\ldots} \left( (\xi_j^k)^{ABC\ldots} \otimes
\varrho^{A'B'C'\ldots} \right) \right] \ge 0$ for all $\xi_j^k$
and $\varrho^{A'B'C'\ldots}$.

The criterion provided by Proposition~\ref{proposition-1} is not
quite operational. To overcome this obstacle we derive sufficient
criteria of entanglement dissociation.

Consider a particular partition $\cP_j^k$. Suppose a linear map
$\Xi:\cT(\cH_{\rm in})\mapsto\cT(\cH_{\rm in})$ which transforms
the density operator $\varrho_{\rm in}$ into some Hermitian (but
not necessarily positive) operator $\Xi[\varrho_{\rm in}]$ such
that
\begin{eqnarray}
\label{requirement} && \!\!\!\! \Big\langle \psi_1^{[\cP_j^k]_1}
\otimes \cdots \otimes \psi_{m-1}^{[\cP_j^k]_{m-1}} \otimes I
\otimes \psi_{m+1}^{[\cP_j^k]_{m+1}} \otimes \cdots \otimes
\psi_k^{[\cP_j^k]_k} \Big| \nonumber\\
&& \times \,\, \Xi[\varrho_{\rm in}] \,\, \Big|
\psi_1^{[\cP_j^k]_1} \otimes \cdots \otimes
\psi_{m-1}^{[\cP_j^k]_{m-1}} \otimes I \otimes
\psi_{m+1}^{[\cP_j^k]_{m+1}} \nonumber\\
&& \quad\quad\quad\quad \otimes \cdots \otimes
\psi_k^{[\cP_j^k]_k} \Big\rangle \ge 0
\end{eqnarray}

\noindent is fulfilled for some vectors
$\psi_1,\ldots,\psi_{m-1},\psi_{m+1},\ldots,\psi_k$, i.e.
$\Xi[\varrho_{\rm in}]$ after projection onto these vectors
becomes a positive operator from the cone
$\cS(\cH^{[\cP_j^k]_m})$. If this is the case, then for rank-1
Kraus operators $A_n \propto \ket{\varphi_n} \bra{\psi_n}$ with
arbitrary $\ket{\varphi_n}$, the operator $(A_1 \otimes \cdots
\otimes A_{m-1} \otimes I \otimes A_{m+1} \otimes \cdots \otimes
A_k) \Xi[\varrho_{\rm in}] (A_1^{\dag} \otimes \cdots \otimes
A_{m-1}^{\dag} \otimes I \otimes A_{m+1}^{\dag} \otimes \cdots
\otimes A_k^{\dag})$ belongs to a cone of separable states
$\sigma_j^k$.

Thus, we obtain the following sufficient criterion of entanglement
dissociation.

\begin{proposition} \label{proposition-2}
Concatenation of a linear Hermitian map $\Xi$ and a
$(k-1)$-partite EB operation $\Big( \cO_{\rm EB}^{[\cP_j^k]_1}
\otimes \cdots \otimes {\rm Id}^{[\cP_j^k]_m} \otimes \cdots
\otimes \cO_{\rm EB}^{[\cP_j^k]_k} \Big)$ belongs to
$\cE[\cD_j^k(\varrho_{\rm in})]$ if $\Xi[\varrho_{\rm in}]$
becomes positive after projection on right-singular vectors of the
rank-1 Kraus operators of the EB operation.
\end{proposition}

The idea of Proposition \ref{proposition-2} is shown for a 6-body
system in Fig. \ref{figure3}. The benefit of the constructed
concatenation is that the map $\Xi$ does not have to be positive
\footnote{A linear map is called positive if it maps positive
operators into positive ones.} (in contrast to Ref.
\cite{filippov-ziman-2013}), which makes the set
$\cE[\cD_j^k(\varrho_{\rm in})]$ even larger.

When all possible states $\varrho_{\rm in}$ are considered, the
satisfaction of requirement (\ref{requirement}) becomes equivalent
to the positivity of the map $\Big( \cO_{\rm EB}^{[\cP_j^k]_1}
\otimes \cdots \otimes {\rm Id}^{[\cP_j^k]_m} \otimes \cdots
\otimes \cO_{\rm EB}^{[\cP_j^k]_k} \Big) \circ \Xi$. This map is
automatically positive if $\Xi$ transforms density operators into
block-positive operators $\xi_j^k$, which in turn is equivalent to
the fact that its Choi operator is block-positive of the form
$\Omega_{\Xi}^{ABC\ldots A'B'C'\ldots} =
\xi^{\cP_j^k(ABC\ldots)|A'B'C'\ldots}$.

\begin{corollary}
\label{corollary-1} If $\Omega_{\Xi}^{ABC\ldots A'B'C'\ldots}$ is
block-positive with respect to the partition
$\cP_j^k(ABC\ldots)|A'B'C'\ldots$, then $\Big( \cO_{\rm
EB}^{[\cP_j^k]_1} \otimes \cdots \otimes {\rm Id}^{[\cP_j^k]_m}
\otimes \cdots \otimes \cO_{\rm EB}^{[\cP_j^k]_k} \Big) \circ \Xi
\in \cD_j^k$ for arbitrary EB operations.
\end{corollary}

\begin{table*}
\caption{\label{table-local} Local depolarizing $N$-qubit channel
$\Phi_q^{\rm local}$: ranges of parameter $q$, for which the
various entanglement-dissociative behaviors are detected (within
the interval $[-\frac{1}{3}, 1]$).}
\begin{ruledtabular}
\begin{tabular}{c|c|c|c|c|c|c|c|c|c}
 $N$ & $\varrho_{\rm in}$ & EA & $\frac{N}{2}\text{Sep} \bigcap 2\text{Ent}$  & $(\frac{N}{2}\!+\!1)\text{Sep} \bigcap \frac{N}{2}\text{Ent}$ & $
   2\text{Sep} \bigcap \frac{N}{2}\text{Ent}$ & $(N\!-\!1)\text{Ent}$=DGE & Not DGE & NPT$(1,N\!-\!1)$ & NPT$(\frac{N}{2},\frac{N}{2})$ \\[1mm] \hline
3 & $\ket{GHZ}$ & $\leqslant$0.490 & $-$ &  $-$ &  $-$ & $\leqslant$0.713 & $>$0.716\footnotemark[1]\footnotemark[2] & $>$0.557 & $-$   \\
  & $\ket{W}$ & $\leqslant$0.485 & $-$ &  $-$ &  $-$ & $\leqslant$0.686 & $>$0.772\footnotemark[1] & $>$0.576 &  $-$   \\
  & $\varrho_{\rm UPB}$ & $\leqslant$0.698 & $-$ & $-$  &  $-$ & $\leqslant$0.852 & $\varnothing$ & $\varnothing$ & $-$  \\
  & all & $\leqslant$0.477 & $-$  & $-$ & $-$ & $\leqslant$0.650 & $-$ &  $-$ & $-$
\\ \hline
4 & $\ket{GHZ}$ & $\leqslant$0.453 & $\leqslant$0.548 & $\leqslant$0.553  & $\leqslant$0.548  & $\leqslant$0.751 & $>$0.781\footnotemark[1]\footnotemark[2] & $>$0.578 & $>$0.512  \\
  & $\ket{W}$ & $\leqslant$0.447 & $\leqslant$0.473 & $\leqslant$0.581 &  $\leqslant$0.473 & $\leqslant$0.756 & $>$0.842\footnotemark[1] & $>$0.585 & $>$0.548  \\
  & $\ket{Cl}$ & $\leqslant$0.444 & $\leqslant$0.478 & $\leqslant$0.574 & $\leqslant$0.478  & $\leqslant$0.742 & $>$0.774\footnotemark[1] & $>$0.532 & $>$0.550  \\
  & all & $\leqslant$0.444 & $\leqslant$0.472 & $\leqslant$0.550 & $\leqslant$0.472 & $\leqslant$0.715 & $-$ & $-$ & $-$  \\
\hline 6 & $\ket{GHZ}$ & $\leqslant$0.414 & $\leqslant$0.433 &
$\leqslant$0.591  & $\leqslant$0.530 & $\leqslant$0.826 & $>$0.850\footnotemark[2] & $>$0.638 & $>$0.490  \\
\end{tabular}
\end{ruledtabular}
\footnotetext[1]{Computation via the method of Ref.
\cite{jungnitsch-2011}.} \footnotetext[2]{Computation via the
method of Ref. \cite{seevinck-2008}}
\end{table*}

The sets $\cE[k\text{Sep}(\varrho_{\rm in})]$ and
$\cE[r\text{Ent}(\varrho_{\rm in})]$ are nothing else but
appropriate convex hulls of sets $\cE[\cD_j^k(\varrho_{\rm in})]$
which can be detected by Proposition \ref{proposition-2}. Let us
remember, however, that we are interested in characterizing sets
$k\text{Sep}(\varrho_{\rm in})$ and $r\text{Ent}(\varrho_{\rm
in})$ of physical (CPT) maps. Since the map $\Phi$ under
investigation is originally CPT, its decomposition into
mathematical maps of the above propositions does not change this
fact but ensures that it belongs to a desired set of maps.
Therefore, we have the following statement.
\begin{proposition}
\label{proposition-3} Suppose a quantum channel $\Phi$ can be
decomposed into the sum $\Phi=\sum_{\cP_j^k \in {\sf P}} \cM_j^k$,
where each elementary map $\cM_j^k \in \cE[\cD_j^k(\varrho_{\rm
in})]$ is constructed via Proposition \ref{proposition-2}. If
${\sf P}$ is a subset of partitions contributing to $k$-separable
or $r$-entangled states, then $\Phi$ belongs to $k{\rm
Sep}(\varrho_{\rm in})$ or $r{\rm Ent}(\varrho_{\rm in})$,
respectively.
\end{proposition}

Similarly, to detect maps from the state-independent sets $k{\rm
Sep}$ and $r{\rm Ent}$ one can use Corollary \ref{corollary-1}
instead of Proposition \ref{proposition-2} in the statement of
Proposition \ref{proposition-3}.

\section{\label{section-application} Applicability of criteria to depolarizing channels}

The sufficient criterion to detect $k{\rm Sep}(\varrho_{\rm in})$
and $r{\rm Ent}(\varrho_{\rm in})$ channels, Proposition
\ref{proposition-3}, implies the existence of the specific
decomposition of the channel of interest, $\Phi$. In this section,
we provide a recipe for construction of such a decomposition for
relatively simple one-parametric families of channels $\Phi$.
Although we do not raise the question of optimality, our findings
enable us to reveal features of the entanglement structure
dynamics.

A general depolarizing map $\Phi: \cT(\cH_d) \mapsto \cT(\cH_d)$
is given by the formula $\Phi = q {\rm Id} + (1-q) {\rm Tr}$,
where ${\rm Tr} [X] = \tr{X} \frac{1}{d} I_d$ is the tracing map.
The map $\Phi$ represents a valid channel (CPT map) if
$q\in[-(d^2-1)^{-1},1]$. Let us consider two one-parametric
families of channels acting on $N$ qubits: the local depolarizing
noise $\Phi_q^{\text{local}} \equiv \Phi_q^{\otimes N}$, where
$\Phi_q$ is a single-qubit map ($d=2$), and the global
depolarizing noise $\Phi_q^{\text{global}}$ ($d=2^N$). Our goal is
the following: for fixed $k$ and $r$, find the region of parameter
$q$ such that the channel $\Phi_q^{\text{local}}$ (or
$\Phi_q^{\text{global}}$) surely adopts the decomposition into
elementary blocks constituting $k\text{Sep} \cap r\text{Ent}$.

In what follows, we do not restrict the number of qubits $N$ but,
in view of the enormous number of possible partitions, we consider
the most interesting cases. All of them represent channels
dissociating genuine entanglement but correspond to various
structures of output states:
\begin{itemize}
\item[(a)] $k=N$ and $r=1$, the output state is fully separable
(EA channels);

\item[(b)] $k=\frac{N}{2}$ and $r=2$, the output state
entanglement mixture is composed of pairs of entangled particles;

\item[(c)] $k=\frac{N}{2}+1$ and $r=\frac{N}{2}$, the biggest
clusters in the output state entanglement mixture cannot contain
more than $\frac{N}{2}$ particles, with the remaining
$\frac{N}{2}$ particles being disentangled;

\item[(d)] $k=2$ and $r=\frac{N}{2}$, the output state
entanglement contains mixtures of two or more clusters of maximum
size $\frac{N}{2}$;

\item[(e)] $k=2$ and $r=N-1$, at least one particle is separated
from entanglement compounds in the output state entanglement
mixture (the biggest subset of DGE channels).
\end{itemize}

For $N=6$ the elementary blocks of these kinds of channels are illustrated
in Fig. \ref{figure3}.

Since the depolarizing channels under investigation are
permutationally invariant, we also consider all possible
permutations of elementary blocks. This is equivalent to
relabelling of particles and, therefore, leads to a simplification
of the analysis of permutationally invariant input states.

To anticipate the results, in Tables \ref{table-local} and
\ref{table-global} we present the ranges of parameter $q$ for
which the depolarizing channels $\Phi_q^{\text{local}}$ and
$\Phi_q^{\text{global}}$, respectively, fall into one of the
classes (a)--(e). Within these ranges, the existence of a
corresponding decomposition in the statement of Proposition
\ref{proposition-3} can be shown [we sum up technical details for
each class (a)--(e) in the forthcoming subsections of the same
label]. The column ``Not DGE'' in Tables \ref{table-local} and
\ref{table-global} is based on detection of geunine entanglement
according to Refs. \cite{jungnitsch-2011,seevinck-2008}. The last
two columns in Tables \ref{table-local} and \ref{table-global} are
based on the conventional negativity under partial transpose (NPT)
entanglement criterion for most asymmetric bipartition (1 body vs.
$N-1$ bodies) and symmetric bipartition ($\frac{N}{2}$ bodies vs.
$\frac{N}{2}$ bodies). In the following subsections A--E, we
present algebra leading to the parameters $q$ for the classes of
channels (a)--(e) above.

\begin{table*}
\caption{\label{table-global} Global depolarizing $N$-qubit
channel $\Phi_q^{\rm global}$: ranges of parameter $q$, for which
the various entanglement-dissociative behaviors are detected
(within the interval $[-(2^{2N}-1)^{-1}, 1]$).}
\begin{ruledtabular}
\begin{tabular}{c|c|c|c|c|c|c|c|c|c}
 $N$ & $\varrho_{\rm in}$ & EA & $\frac{N}{2}\text{Sep} \bigcap 2\text{Ent}$  & $(\frac{N}{2}\!+\!1)\text{Sep} \bigcap \frac{N}{2}\text{Ent}$ & $
   2\text{Sep} \bigcap \frac{N}{2}\text{Ent}$ & $(N\!-\!1)\text{Ent}$=DGE & Not DGE & NPT$(1,N\!-\!1)$ & NPT$(\frac{N}{2},\frac{N}{2})$ \\[1mm] \hline
3 & $\ket{GHZ}$ & $\leqslant$0.147 & $-$ & $-$ & $-$ & $\leqslant$0.402 & $>$0.429\footnotemark[3]\footnotemark[4] & $>$0.200 & $-$  \\
  & $\ket{W}$ & $\leqslant$0.125 & $-$ & $-$ & $-$ & $\leqslant$0.317 & $>$0.479\footnotemark[3] &  $>$0.210 & $-$ \\
  & $\varrho_{\rm UPB}$ & $\leqslant$0.400 & $-$ & $-$ & $-$ & $\leqslant$0.690 & $\varnothing$ & $\varnothing$ & $-$ \\
  & all & $\leqslant$0.111 & $-$ & $-$ & $-$ & $\leqslant$0.289 & $-$ & $-$ &  $-$ \\
\hline
4 & $\ket{GHZ}$ & $\leqslant$0.062 & $\leqslant$0.202 & $\leqslant$0.111 & $\leqslant$0.202  & $\leqslant$0.262 & $>$0.467\footnotemark[3]\footnotemark[4] & $>$0.112 & $>$0.112    \\
  & $\ket{W}$ & $\leqslant$0.048 & $\leqslant$0.123 & $\leqslant$0.124 & $\leqslant$0.123  & $\leqslant$0.256 & $>$0.474\footnotemark[3] &  $>$0.127 & $>$0.112    \\
  & $\ket{Cl}$ & $\leqslant$0.052 & $\leqslant$0.123 & $\leqslant$0.109 & $\leqslant$0.123  & $\leqslant$0.229 & $>$0.385\footnotemark[3] &  $>$0.112 & $>$0.112    \\
  & all & $\leqslant$0.047 & $\leqslant$0.121 & $\leqslant$0.107 & $\leqslant$0.121 & $\leqslant$0.184 & $-$ & $-$ & $-$    \\
\hline 6 & $\ket{GHZ}$ & $\leqslant$0.011 & $\leqslant$0.034 &
$\leqslant$0.032 & $\leqslant$0.046 & $\leqslant$0.131 &
$>$0.493\footnotemark[4]
& $>$0.031 & $>$0.031 \\
\end{tabular}
\end{ruledtabular}
\footnotetext[3]{Computation via the method of Ref.
\cite{jungnitsch-2011}.} \footnotetext[4]{Computation via the
method of Ref. \cite{seevinck-2008}}
\end{table*}

\subsection{\label{section-applicability-ea} Entanglement annihilating channels}

The elementary block of EA channel is obtained by applying
entanglement breaking operations on $N-1$ particles (see Fig.
\ref{figure3}a). The exact form of chosen EB operations reads
$\cO_{{\rm EB}\psi_i}[X] = \frac{1}{2} \ket{\psi_i}\bra{\psi_i} X
\ket{\psi_i}\bra{\psi_i}$, where $\{ \frac{1}{2}
\ket{\psi_i}\bra{\psi_i} \}_{i=1}^4$ form a symmetric
informationally complete positive operator-valued measure
(SIC-POVM) for qubits (see the explicit analytical form of the
vectors $\{\ket{\psi_i}\}_{i=1}^{4}$ in
~\cite{scott-grassl-2010}). This choice is justified by the fact
that $\sum_{i=1}^4 \cO_{{\rm EB}\psi_i} = \Phi_{q=1/3}$. (The same
result would be obtained by using projectors on mutually unbiased
bases \cite{wootters-fields-1989} instead of SIC-POVM elements,
however, this approach leads to worse results for some input
states $\varrho_{\rm in}$.) The suggested decomposition reads
\begin{equation}
\label{EA-resolution} \Phi = \frac{1}{N} \sum_{m=1}^N \Big(
\Phi_{q=1/3}^{[\cP^N]_1} \otimes \cdots \otimes {\rm
Id}^{[\cP^N]_m}  \otimes \cdots \otimes \Phi_{q=1/3}^{[\cP^N]_N}
\Big) \circ \Xi_{a}(m),
\end{equation}

\noindent where $m$ is the index of a particle not subjected to EB
operations. We have taken into account that each $\Phi_{q=1/3}$ is
composed of EB operations $\cO_{{\rm EB}\psi_i}$ and therefore it
is convenient to parameterize the map $\Xi_{a}$ in such a way that
the vectors $\ket{\psi_i}$ are not included in the parametrization
directly. However, the linear map $\Xi_{a}(m)$ should satisfy the
requirement (\ref{requirement}) for all choices of vectors
$\ket{\psi_{i_t}}_{t=1}^{N-1}$ from the set
$\{\ket{\psi_i}\}_{i=1}^{4}$. To parameterize the map $\Xi_a(m)$
we resort to a so-called diagonal map of the form
\begin{eqnarray}
\label{diagonal} \Xi[X] &=& \frac{1}{2^N}
\sum_{i_1,\ldots,i_N=0,\ldots,3} x_{i_1 \cdots i_N}
\tr{(\varsigma_{i_1}\otimes\cdots\otimes\varsigma_{i_N})X} \nonumber\\
&& \qquad \qquad \qquad \qquad \qquad \times
\varsigma_{i_1}\otimes\cdots\otimes\varsigma_{i_N},
\end{eqnarray}

\noindent where $\varsigma_0=I_2$ and $\varsigma_1$,
$\varsigma_2$, $\varsigma_3$ are conventional Pauli matrices. Let
$\#_0[i_1 \cdots i_N]$ denote the number of zeros in the sequence
$i_1,\ldots,i_N$. Consider diagonal maps $\Xi_{a}(m)$ such that
the coefficients $\{x_{i_1 \cdots i_N}\}$ depend on $\#_0[i_m]$
and $\#_0[i_1 \cdots i_{m-1} i_{m+1} \cdots i_N]$ only, i.e.
$x_{i_1 \cdots i_N} = f_{a}(\#_0[i_m],\#_0[i_1 \cdots i_{m-1}
i_{m+1} \cdots i_N])$, with restrictions on the parameters
$\{f_a\}$ being imposed by (\ref{requirement}). Then, the relation
(\ref{EA-resolution}) becomes valid if
\begin{eqnarray}
\label{EA-system} && \frac{n}{3^{n-1}N} f_{a}(0,N-n) +
\frac{N-n}{3^{n} N}
f_{a}(1,N-n-1) \nonumber\\
&& = \left\{
\begin{array}{ll}
  q^n, & \Phi=\Phi_q^{\text{local}}, \\
  q^{1-\delta_{n,0}}, & \Phi=\Phi_q^{\text{global}}, \\
\end{array}\right. \qquad n=0,\ldots,N,
\end{eqnarray}

\noindent where $\delta_{s,t}$ is the conventional Kronecker
delta.

For a fixed input state $\varrho_{\rm in}$, we find the
restrictions on the parameters $\{f_a\}$ given by
(\ref{requirement}) and then solve the system of equations
(\ref{EA-system}) numerically. If the system has a solution for
some $\tilde{q}$, then it also has a solution for $q<\tilde{q}$.
Solutions ($\max \tilde{q}$) are presented for some interesting
states $\varrho_{\rm in}$ \footnote{The states of interest are
$\ket{GHZ} = \frac{1}{\sqrt{2}}(\ket{0}^{\otimes
N}+\ket{1}^{\otimes N})$; $\ket{W} =
\frac{1}{\sqrt{N}}(\ket{10\ldots 0}+\ket{01\ldots
0}+\cdots+\ket{00\ldots 1})$; $\varrho_{\rm UPB} = \frac{1}{4}(I_8
- P_{\rm UPB})$, where $P_{\rm UPB}$ is a projector on
unextendible product bases for 3 qubits
\cite{bennet-1999,divincenzo-2003}; $\ket{Cl} =
\frac{1}{2}(\ket{0000}+\ket{0011}+\ket{1100}-\ket{1111})$.} of
$N=3,4,6$ qubits in Tables \ref{table-local} and
\ref{table-global} for local and global noises, respectively. We
also consider the case of all possible input states as follows:
since $\Xi_a$ linearly depends on the parameters $\{f_a\}$, we
check the corresponding block-positivity of $\Omega_{\Xi_a}$ (see
Corollary \ref{corollary-1}) for some number of parameters
$\{f_a\}$ and construct a convex hull of satisfactory parameters;
then we solve the system of equations (\ref{EA-system}) for
$\{f_a\}$ from the convex hull; the maximum value $q$ for which
the system has a solution is presented in Tables \ref{table-local}
and \ref{table-global} in the rows ``all''.

\subsection{$\frac{N}{2}\text{Sep} \bigcap 2\text{Ent}$ channels}

The output state will be $\frac{N}{2}$-separable and $2$-entangled
if the channel can be decomposed into elementary transformations
$\cE[\cD_j^{N/2}(\varrho_{\rm in})]$ from
Proposition~\ref{proposition-2}, each containing $(\frac{N}{2} - 1
)$ EB operations $\cO_{\text{EB}}$ on two qubits (see Fig.
\ref{figure3}b). As in the previous subsection, we choose EB
operations of the form
$\cO_{\text{EB}\psi_i}[X]=\frac{1}{4}\ket{\psi_i}\bra{\psi_i} X
\ket{\psi_i}\bra{\psi_i}$, where $\{ \frac{1}{4}
\ket{\psi_i}\bra{\psi_i} \}_{i=1}^{16}$ form a SIC-POVM in
$\cT(\cH_4)$ (see the explicit analytical form of the vectors
$\{\ket{\psi_i}\}_{i=1}^{16}$ in ~\cite{scott-grassl-2010}). Then
$\sum_{i=1}^{16} \cO_{{\rm EB}\psi_i}^{AB} = \Phi_{q=1/5}^{AB}$ is
a depolarizing map acting on two qubits ($A$ and $B$)
simultaneously. The decomposition of channel $\Phi$ reads
\begin{eqnarray}
\label{Nover2-separable-resolution} \Phi &=& {\binom{N}{2}}^{-1}
\sum_{\cP_j^{N/2} \in {\sf P}} \, \sum_{m=1}^{N/2} \Big(
\Phi_{q=1/5}^{[\cP_j^{N/2}]_1} \otimes \cdots
\nonumber\\
&& \otimes {\rm Id}^{[\cP_j^{N/2}]_m} \otimes \cdots \otimes
\Phi_{q=1/5}^{[\cP_j^{N/2}]_{N/2}} \Big) \circ \Xi_{b}(j,m),
\end{eqnarray}

\noindent where ${\sf P}$ is a set of $\frac{N!}{2^{N/2}(N/2)!}$
partitions $\cP_j^{N/2}$ such that $\#[\cP_j^{N/2}]_1 = \ldots =
\#[\cP_j^{N/2}]_{N/2} = 2$ (two qubits in each party), and the map
$\Xi_{b}(j,m)$ must meet the condition (\ref{requirement}) for all
choices of vectors $\ket{\psi_{i_t}}_{t=1}^{(N/2)-1}$ from the set
$\{\ket{\psi_i}\}_{i=1}^{16}$. Using diagonal maps $\Xi_{b}(j,m)$
of the form (\ref{diagonal}) with the parametrization $x_{i_1
\cdots i_N} = f_{b}(\#_0[i_l i_{l'}],\#_0[i_1 \cdots i_{l-1}
i_{l+1} \cdots i_{l'-1} i_{l'+1} \cdots i_N])$, $(i_{l}
i_{l'})\in[\cP_j^{N/2}]_m$, we obtain the following system of
equations:
\begin{eqnarray}
&& {\binom{N}{2}}^{-1} \bigg\{ \binom{n}{2}
\frac{f_{b}(0,N-n)}{5^{\lceil n/2 \rceil
- 1}} + n(N-n)\frac{f_{b}(1,N-n-1)}{5^{\lfloor n/2 \rfloor}} \nonumber\\
&& \qquad\qquad + \binom{N-n}{2} \frac{f_{b}(2,N-n-2)}{5^{\lceil
n/2 \rceil}} \bigg\} \nonumber\\
&& = \left\{
\begin{array}{ll}
  q^n, & \Phi=\Phi_q^{\text{local}}, \\
  q^{1-\delta_{n,0}}, & \Phi=\Phi_q^{\text{global}}, \\
\end{array}\right. \qquad n=0,\ldots,N.
\end{eqnarray}

The maximal values of $q$, for which the system has a solution
compatible with (\ref{requirement}), are presented for various
input states in Tables \ref{table-local} and \ref{table-global}.

\subsection{$(\frac{N}{2}+1)\text{Sep} \bigcap \frac{N}{2}\text{Ent}$ channels}

The output state will be $(\frac{N}{2}+1)$-separable and
$\frac{N}{2}$-entangled if the channel can be decomposed into
elementary transformations $\cE[\cD_j^{N/2+1}(\varrho_{\rm in})]$,
$j=1, \ldots, \binom{N}{N/2}$ (for such $j$s, the $N$-body system
is divided into $\frac{N}{2}$ single-body parts plus one part
comprising $\frac{N}{2}$ bodies, see Fig. \ref{figure3}c). To find
the decomposition for Proposition \ref{proposition-3}, we use the
single-qubit EB operations $\cO_{\text{EB}\psi_i}$,
$\ket{\psi_i}\in\cH_2$, $i=1,\ldots,4$ as in Sec.
\ref{section-applicability-ea}. This yields the following
decomposition:
\begin{eqnarray}
\label{Nover2-separable-and-entangled-resolution} \Phi &=&
{\binom{N}{N/2}}^{-1} \sum_{j=1}^{\binom{N}{N/2}} \Big(
\Phi_{q=1/3}^{[\cP_j^{N/2+1}]_1} \otimes \cdots \otimes
\Phi_{q=1/3}^{[\cP_j^{N/2+1}]_{N/2}} \nonumber\\
&& \qquad\qquad\qquad\quad \otimes {\rm
Id}^{[\cP_j^{N/2+1}]_{N/2+1}} \Big) \circ \Xi_{c}(j),
\end{eqnarray}

\noindent where the map $\Xi_{c}(j)$ must satisfy the requirement
(\ref{requirement}) for all choices of vectors
$\ket{\psi_{i_t}}_{t=1}^{N/2}$ from the set
$\{\ket{\psi_i}\}_{i=1}^{4}$ (vectors corresponding to SIC-POVM
for a qubit). Using diagonal maps $\Xi_{c}(j)$ of the form
(\ref{diagonal}) with the parametrization $x_{i_1 \cdots i_N} =
f_{c}(\#_0[\{i_1 \cdots i_N\} \setminus \{i_{l_1} \cdots
i_{l_{N/2}}\}],\#_0[i_{l_1} \cdots i_{l_{N/2}}])$, $(i_{l_1}
\cdots i_{l_{N/2}})\in[\cP_j^{N/2+1}]_{N/2+1}$, we obtain the
following system of equations:
\begin{eqnarray}
&& {\binom{N}{N/2}}^{-1} \sum_{l=0}^{N/2} \binom{n}{N/2-l}
\binom{N-n}{l} \frac{f_{c}(l,N-n-l)}{3^{N/2-l}}  \nonumber\\
&& = \left\{
\begin{array}{ll}
  q^n, & \Phi=\Phi_q^{\text{local}}, \\
  q^{1-\delta_{n,0}}, & \Phi=\Phi_q^{\text{global}}, \\
\end{array}\right. \qquad n=0,\ldots,N.
\end{eqnarray}

The maximal values of $q$, for which the system has a solution
compatible with (\ref{requirement}), are presented for various
input states in Tables \ref{table-local} and \ref{table-global}.

\subsection{$2\text{Sep} \bigcap \frac{N}{2}\text{Ent}$ channels}

The output state will be $2$-separable and $\frac{N}{2}$-entangled
if the channel can be decomposed into elementary transformations
$\cE[\cD_j^2 (\varrho_{\rm in})]$,
$j=\stirling{N}{2}-\frac{1}{2}\binom{N}{N/2}+1,\ldots,
\stirling{N}{2}$ (such choice of $j$s corresponds to bipartitions
of an $N$-body system into equal $\frac{N}{2}$-body parts, see
Fig. \ref{figure3}d). We use the following EB operations
$\cO_{\text{EB}}$ on $\frac{N}{2}$ qubits:
$\cO_{\text{EB}\psi_i}[X]=\frac{1}{2^{N/2}}\ket{\psi_i}\bra{\psi_i}
X \ket{\psi_i}\bra{\psi_i}$, where $\{\ket{\psi_i}\}_{i=1}^{2^N}$
is a set of normalized vectors such that
$\{\ket{\psi_i}\bra{\psi_i}\}_{i=1}^{2^N}$ is a set of SIC
projectors (see the explicit form of vectors
$\{\ket{\psi_i}\}_{i=1}^{2^N}$ up to $N=12$ in
\cite{scott-grassl-2010}). [Let us recall that a particular form
of vectors $\ket{\psi_i}$ is important only for a particular input
state $\varrho_{\rm in}$. If the input state is arbitrary, i.e.
the domain is $\cS(\cH_2^{\otimes N})$, then one should not care
about the specific form of EB operations.] The important fact is
that $\sum_{i=1, \ldots, 2^N} \cO_{\text{EB}\psi_i} =
\Phi_{q=(2^{N/2}+1)^{-1}}$ is a depolarizing map acting on
$\frac{N}{2}$ qubits. The possible decomposition reads
\begin{eqnarray}
\label{Nover2-entangled-resolution} \Phi &=& {\binom{N}{N/2}}^{-1}
\sum_{j=\stirling{N}{2}-\frac{1}{2}\binom{N}{N/2}+1}^{\stirling{N}{2}}
\, \sum_{m=1}^2
\nonumber\\
&& \Big( \Phi_{q=(2^{N/2}+1)^{-1}}^{[\cP_j^{2}]_m} \otimes {\rm
Id}^{[\cP_j^{2}]_{\{1,2\} \setminus m}} \Big) \circ \Xi_{d}(j,m),
\end{eqnarray}

\noindent where $\Xi_{d}(j,m)$ must satisfy condition
(\ref{requirement}) for all vectors $\{\ket{\psi_i}\}_{i=1}^{2^N}$
and for the corresponding domain of density operators
$\varrho_{\rm in}$. For computational reasons let us note that
checking the validity of (\ref{requirement}) is less
time-consuming when we justify the positivity of the Hermitian
operator without revealing its eigenvalues. Namely, the
eigenvalues of a $d \times d$ Hermitian matrix $X$ are
non-negative if and only if $C_k \ge 0$ for $k=1,\ldots,d$, where
$C_k$ is given by the recurrence relation $C_{k} = \frac{1}{k}
\sum_{l=1}^{k} (-1)^{l-1} C_{k-l} \tr{X^l}$ with initial condition
$C_0=1$ \cite{bengtsson-zyczkowski-2006}. We use this technique
for $N=6$.

Using diagonal maps $\Xi_{d}(j,m)$ of the form (\ref{diagonal})
with the parametrization $x_{i_1 \cdots i_N} = f_{d}(\#_0[i_{l_1}
\cdots i_{l_{N/2}}],\#_0[\{i_1 \cdots i_N\} \setminus \{i_{l_1}
\cdots i_{l_{N/2}}\}])$, $(i_{l_1} \cdots
i_{l_{N/2}})\in[\cP_j^2]_m$, we obtain the following system of
equations:
\begin{eqnarray}
&& {\binom{N}{N/2}}^{-1} \sum_{l=0}^{N/2} \binom{n}{N/2-l}
\binom{N-n}{l} \frac{f_{d}(l,N-n-l)}{(2^{N/2}+1)^{1-\delta_{l,N/2}}}  \nonumber\\
&& = \left\{
\begin{array}{ll}
  q^n, & \Phi=\Phi_q^{\text{local}}, \\
  q^{1-\delta_{n,0}}, & \Phi=\Phi_q^{\text{global}}, \\
\end{array}\right. n=0,\ldots,N.
\end{eqnarray}

The maximal values of $q$, for which the system has a solution,
are presented for various $\varrho_{\rm in}$ in Tables
\ref{table-local} and \ref{table-global}.

\begin{figure}[b]
\includegraphics[width=8.5cm]{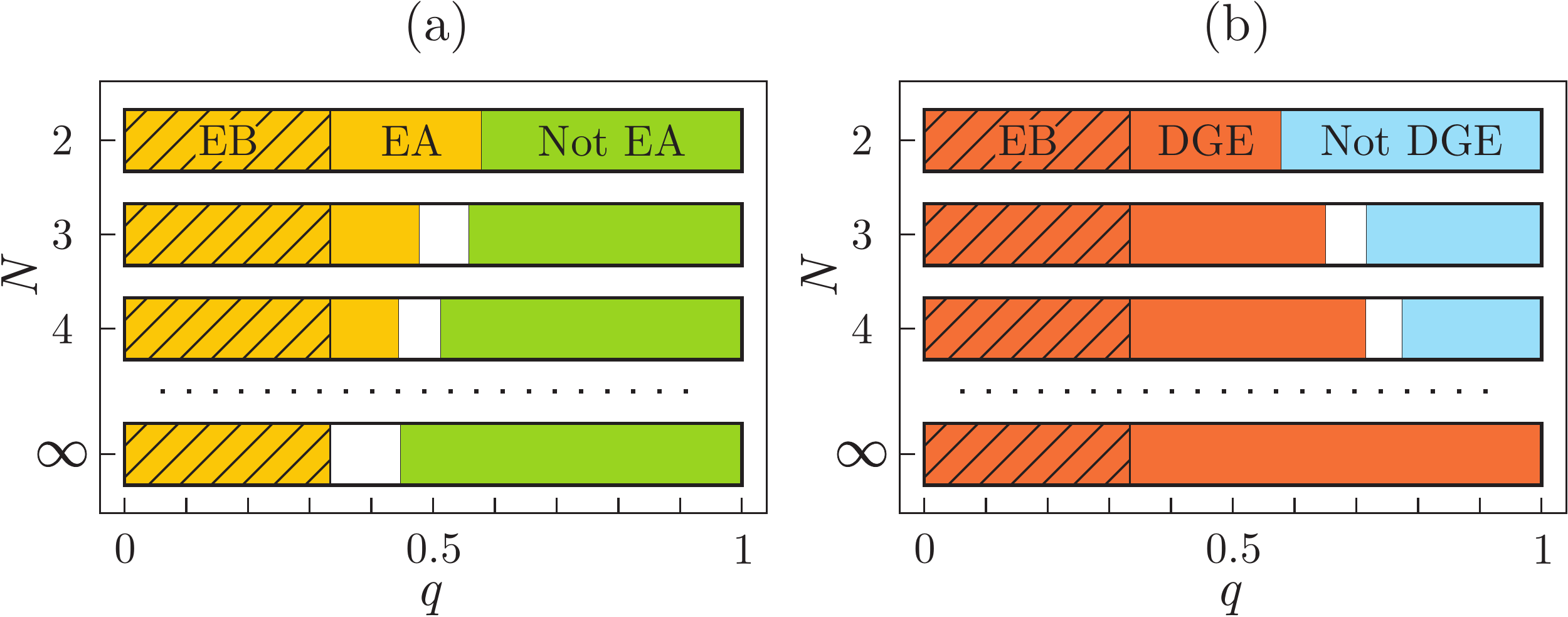}
\caption{\label{figure4} Scaling of the entanglement degradation
properties of an $N$-qubit local depolarizing channel
$\Phi_q^{\otimes N}$ with increasing $N$: (a) entanglement
annihilation, (b) dissociation of genuine entanglement.}
\end{figure}

\begin{figure*}
\includegraphics[width=18cm]{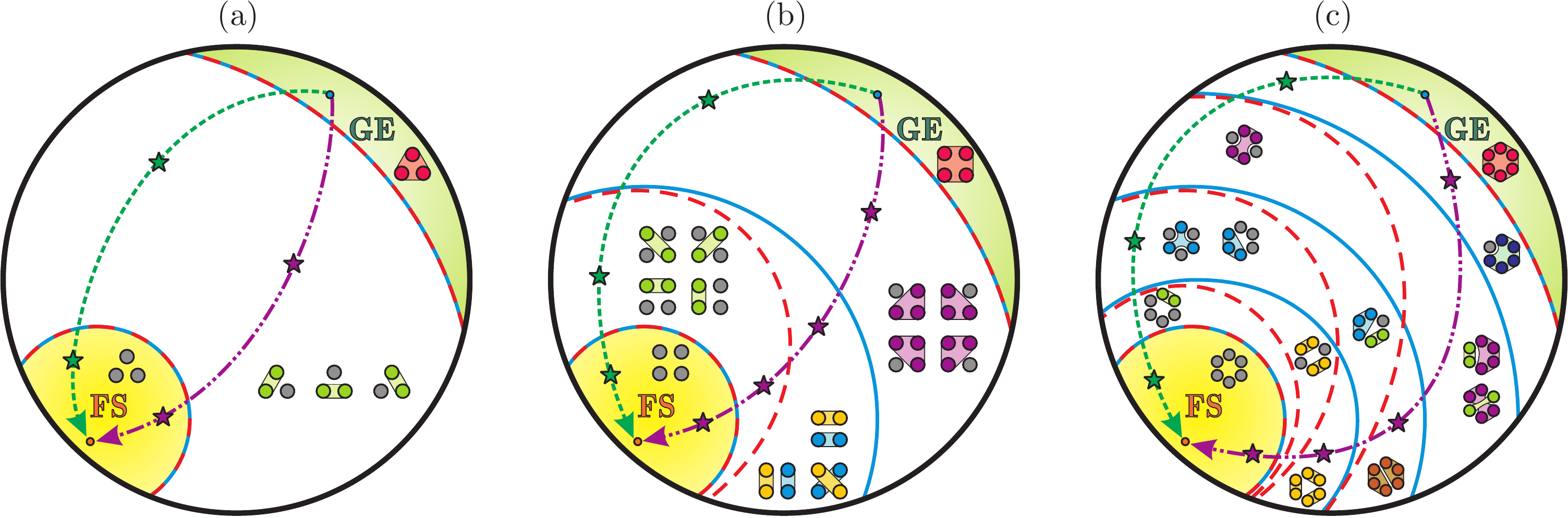}
\caption{\label{figure5} Tracks of the typical entanglement
structure dynamics subject to local depolarizing noise (green
dotted line) and global depolarizing noise (purple dash-dotted
line) for $N$-qubit systems: (a) $N=3$, (b) $N=4$, (c) $N=6$. The
state space $\cS(\cH_2^{\otimes N})$ is divided into areas of
$k$-separable states (red dashed lines) and $r$-entangled states
(blue solid lines), with representatives of the states being
depicted. Stars on the tracks denote points detected in Sec.
\ref{section-application} and listed in Tables \ref{table-local}
and \ref{table-global}.}
\end{figure*}

\subsection{\label{section-applicability-age} Channels dissociating genuine entanglement}

The elementary blocks of these channels can be obtained by
applying an EB operation on a single qubit (Fig. \ref{figure3}e).
We use the same EB operations as in subsection
\ref{section-applicability-ea}. This yields the decomposition
\begin{equation}
\label{AGE-resolution} \Phi = \frac{1}{N} \sum_{m=1}^N \Big(
\Phi_{q=1/3}^{[\cP_{j=m}^2]_1} \otimes {\rm Id}^{[\cP_{j=m}^2]_2}
\Big) \circ \Xi_{e}(m),
\end{equation}

\noindent where the map $\Xi_{e}(m)$ must satisfy
(\ref{requirement}) for all vectors $\{\ket{\psi_i}\}_{i=1}^{4}$
corresponding to SIC-POVM for a qubit. Using diagonal maps
$\Xi_{e}(m)$ of the form (\ref{diagonal}) with the parametrization
$x_{i_1 \cdots i_N} = f_{e}(\#_0[i_m],\#_0[i_1 \cdots i_{m-1}
i_{m+1} \cdots i_N])$, we find that (\ref{AGE-resolution}) becomes
a valid equality if
\begin{eqnarray}
\label{age-system} && \frac{n}{3N} f_{e}(0,N-n) + \frac{N-n}{N}
f_{e}(1,N-n-1)
\nonumber\\
&& = \left\{
\begin{array}{ll}
  q^n, & \Phi=\Phi_q^{\text{local}}, \\
  q^{1-\delta_{n,0}}, & \Phi=\Phi_q^{\text{global}} \\
\end{array}\right. \qquad n=0,1,\ldots,N. \qquad
\end{eqnarray}

\noindent The maximal values $q$, for which the system has a
solution and (\ref{requirement}) is fulfilled, are presented for
various input states in Tables \ref{table-local} and
\ref{table-global}.

\section{\label{section-discussion}Discussion}

To begin with, the NPT criterion gives a little information about
the multipartite entanglement structure. Indeed, one can observe
in Tables \ref{table-local} and \ref{table-global} many situations
when $\Phi[\varrho_{\rm in}]$ is either negative under partial
transpose but not genuinely entangled, or positive under partial
transpose but not fully separable. The gap between states
$\Phi[\varrho_{\rm in}]$ that are surely not genuinely entangled
and those that are definitely genuinely entangled is quite narrow
for particular input states. This can be treated as an indication
of the efficiency of the rather simple decomposition
(\ref{AGE-resolution}) involving single-qubit entanglement
breaking operations.

The remarkable fact is that our method enables us to consider
\textit{all} input states and find channels that transform any of
them to a particular entanglement structure. This is what we mean
by a ``typical'' behavior. For example, we can detect channels
that annihilate entanglement, for which the output state is always
fully separable whatever the input state is. Note, that the bounds
obtained on $q$ for entanglement annihilation are higher than
those that can be found via the condition of sufficiently small
purity $\tr{(\Phi[\varrho_{\rm in}])^2}$ \cite{hildebrand-2007}.
Scaling of EA for the local depolarizing channel is shown in Fig.
\ref{figure4}a. When $N \rightarrow \infty$, the channel
$\Phi_q^{\otimes N}$ cannot be EA if $q>\frac{1}{\sqrt{5}}$
\cite{simon-2002}, however, the question if $\text{EA}=\text{EB}$
still remains an open problem. On the contrary, from formula
(\ref{age-system}) one can see that the genuine entanglement of
any input state can be dissociated by a negligible noise in the
limit $N \rightarrow \infty$ (Fig. \ref{figure4}b).

The dissipative dynamics under consideration can be described by
the gradually decreasing parameter $q \sim e^{-\Gamma t}$, where
the dissipation rate $\Gamma$ takes, in principle, different
values for local and global noises. For our purposes it is enough
to know that $q$ continuously diminishes from $q=1$ to $q=0$. For
such types of dissipative dynamics, the state evolution through
the nested sets of Sec. \ref{section-formalism} is irreducible:
once the state comes into a particular ``doll'' of the structure,
it cannot escape it in the future.

Using the data from Table \ref{table-local}, we may conclude that
the dissociation of genuine multiparticle entanglement under local
depolarizing noise starts by detaching a single random particle
(i.e. the state becomes $(N-1)$-entangled). Then the noise
detaches particles one by one resulting in $k$-separable
$(N-k+1)$-entangled states ($k$ increases with decreasing $q$).
Indeed, since the noise is local, once a particle is detached from
the entanglement compound, there is no way for it to rejoin (Fig.
\ref{figure5}). Finally, the noisy evolution makes the state fully
separable.

The analysis of Table \ref{table-global} shows that the
entanglement dissociation progresses in a different way under
global depolarizing noise: while the beginning stage also implies
detaching of a single random particle from the entanglement
compound, in further dynamics this particle can fuse with another
one and form a two-particle entanglement cluster that is detached
from the main compound (a convex combination of such states). The
process continues until the point when the original compound is
divided into two clusters ($2$-separable $\frac{N}{2}$-entangled
state), then the detachment of particles and their successive
fusion result in the formation of more entanglement clusters of
smaller size ($k$-separable $\frac{N}{k}$-entangled state, $k$
increases with decreasing $q$), and so on until the full
separability (see Fig. \ref{figure5} for the case of 6 qubits).

\section{\label{section-summary} Summary}

Our study was motivated by the necessity to know the multiparticle
entanglement structure and its vulnerability to noises in physical
and quantum-informational applications. We did not restrict
ourselves to specific input states and considered the set of all
possible states as well. We found criteria for maps dissociating
entanglement with respect to a particular partition and developed
sufficient conditions for their reliable detection. Namely, the
channel of interest should adopt a decomposition into (not
necessarily completely positive) linear maps which give rise to
the desired form of the output. One can draw a rough analogy
between this decomposition and the path integral formulation of
quantum mechanics, where the trajectories can be quite
non-physical but this does not affect the resulting physical
evolution. For local and global depolarizing $N$-qubit channels we
provided a simple strategy of constructing decompositions that
allowed us to find noise levels guaranteeing the particular form
of entanglement structure. Our decompositions are not optimal and
can in principle be improved by applying modifications of
semidefinite programming \cite{jungnitsch-2011} and other
algorithms \cite{kampermann-2012} for Choi operators.
Nevertheless, our toolbox allowed us to reveal differences in
entanglement structure dynamics under local and global noises: the
particles split one by one from the entanglement compound in the
case of local noise, and tend to form clusters in the case of
global noise. We believe that the obtained results may be extended
to other noise models and provide additional information about the
general rules of the dynamics of multiparticle entanglement
structure.

\textit{Acknowledgments}. This work was supported by EU integrated
project SIQS, COST Action MP1006, APVV-0646-10 (COQI) and VEGA
2/0127/11 (TEQUDE). S.N.F. acknowledges support from the National
Scholarship Programme of the Slovak Republic and the 7th FP
project iQIT. S.N.F. and A.A.M. acknowledge support from the
Dynasty Foundation and the Russian Foundation for Basic Research
under Project No. 12-02-31524-mol-a. A.A.M. acknowledges partial
support from the Austrian Science Fund (FWF) through the SFB
FoQuS: F 4012. M.Z. acknowledges support from 7th FP STREP project
RAQUEL and GACR project P202/12/1142.

\end{document}